\documentclass{aamas2017}

\newcommand{\mytitle}{Proportional Representation in Vote Streams}

\usepackage{ntheorem}

\newtheorem{thm}{Theorem}

\usepackage{thm-restate}
\usepackage{etoolbox}
\usepackage{paralist}
\usepackage{multirow}
\usepackage{xspace}
\usepackage{graphicx}
\usepackage{tikz}
\usepackage[textsize=scriptsize]{todonotes}
\usepackage{hyperref}
\usepackage[noend,ruled]{algorithm2e}
\usepackage{colortbl}
\usepackage{lscape}
\usepackage{comment}
\usepackage{float}
\usepackage{tikz}
\usepackage{mathtools}
\usetikzlibrary{positioning}
\usepackage{wasysym}

\newtheorem{theorem}{Theorem}

\newtheorem{definition}{Definition}

\newcommand{\calR}{\mathcal{R}}

\newcommand{\mypara}[1]{\medskip\noindent\textbf{#1.}}

\newcommand{\pref}{\ensuremath{\succ}}
\newcommand{\pos}{\mathrm{pos}}
\newcommand{\score}{\mathrm{score}}

\newtheorem{claim}{Claim}

\usepackage{enumitem}
\setlist[itemize]{label=\textbullet,leftmargin=*,itemsep=.05cm}

\title{\mytitle}

\numberofauthors{3}
\author{
\alignauthor       
Palash Dey\\
\affaddr{Tata Institute of Fundamental Research}\\
\affaddr{Mumbai, India}\\
       \email{palash.dey@tifr.res.in}
\alignauthor       
Nimrod Talmon\\
\affaddr{Weizmann Institute of Science}\\
\affaddr{Rehovot, Israel}\\
       \email{nimrodtalmon77@gmail.com}
\alignauthor       
Otniel van Handel\\
\affaddr{Weizmann Institute of Science}\\
\affaddr{Rehovot, Israel}\\
       \email{otnivh@gmail.com}}

\begin{document}

\pagenumbering{arabic}

\maketitle

\begin{abstract}
We consider elections where the voters come one at a time,
in a streaming fashion, and devise space-efficient algorithms which identify an approximate winning committee with respect to common multiwinner proportional representation voting rules;
specifically,
we consider the Approval-based and the Borda-based variants of both the Chamberlin-- \linebreak Courant rule and the Monroe rule.
We complement our algorithms with lower bounds.
Somewhat surprisingly,
our results imply that,
using space which does not depend on the number of voters 
it is possible to efficiently identify an approximate representative committee of fixed size over vote streams
with huge number of voters.
\end{abstract}




\section{Introduction}

The voting rule suggested by Chamberlin--Courant~\cite{cha-cou:j:cc} and the voting rule suggested by Monroe~\cite{mon:j:monroe},
are multiwinner voting rules concentrated on proportional representation.
Such proportional representation multiwinner rules aim at selecting a committee of fixed size which represents the society best.
Informally,
most voters shall be somewhat satisfied by the committees selected by such proportional representation rules,
which,
roughly speaking,
try to best represent the spectrum of different views of the society.
This stands in contrast,
for example,
to $k$-best multiwinner voting rules such as $k$-Borda.
Proportional representation multiwinner voting rules have several good axiomatic properties~\cite{elkind2014properties}.

Winner determination for these rules,
however,
is NP-hard~\cite{pro-ros-zoh:j:proportional-representation},
though it is possible to compute the winner when some parameters are small~\cite{bet-sli-uhl:j:mon-cc};
that is, winner determination for these rules is fixed-parameter tractable with respect to either
the number of voters or the number of candidates.
Further,
efficient approximation algorithms are known~\cite{sko-fal-sli:c:multiwinner}
for these rules as well as heuristic algorithms based on clustering~\cite{clusteringpaper}.

Proportional representation multiwinner voting rules have several other applications,
besides their original, political application.
Specifically,
these rules are used for resource allocation~\cite{mon:j:monroe,sko-fal-sli:c:multiwinner},
facility location~\cite{bet-sli-uhl:j:mon-cc},
and recommender systems~\cite{bou-lu:c:chamberlin-courant,skowron2015finding}.
In such situations,
it is indeed desirable to select a set of $k$ ``representative'' elements out of a larger set.

While the number of voters in some elections is modest,
there are situations where the number of voters is huge,
making it impossible to store the whole election in order to operate upon it
(specifically, to identify a winning committee).
Consider,
e.g.,
the preferences of users of an online shopping website:
  there are lots of potential buyers (corresponding to the voters),
  each with her own preferences over the items being sold on the website (corresponding to the candidates).
The owners of the shopping website might wish to identify a set of, say, $k$ items to advertise on their landing page,
with the intent of maximizing the number of users which would be interested in at least one of those displayed items.

More generally,
as certain tasks which are concerned with the creation of various kinds of product portfolios
can be modeled as equivalents of solving winner determination for proportional representation,
it is of interest to devise efficient algorithms for such situations which naturally correspond to elections with huge number of voters.
Thus,
in this paper we are interested in designing algorithms which identify good representative committees,
but without being able to store the whole electorate in order to process it;
concretely, we are aiming at algorithms whose space complexity does not depend on the number of voters,
since this number might be huge.

To study space-efficient algorithms for such situations,
we consider streaming algorithms which solve the winner determination problem for proportional representation multiwinner voting rules.
Specifically,
while we consider the set of alternatives as being fixed,
we assume that the voters are arriving (that is, voting) one at a time,
in what we refer to as a \emph{vote stream}.
Concretely,
we assume that each voter is arriving only once
(such that it is possible to process each voter only once),
and we are interested in space-efficient streaming algorithms for finding a winning committee of fixed size~$k$ in such vote streams.

As it is customary in studying streaming algorithms,
we allow our algorithms to be randomized and to find approximate solutions.
That is,
in order to have algorithms which use only small amounts of space,
we will be satisfied with algorithms which find an approximate winning committee;
specifically,
we will be satisfied with finding a committee whose score under the given voting rule
is close to being the optimum score possible for a committee with respect to the given election.
Further,
we will be satisfied with randomized algorithms,
which might not always find such approximate winning committees,
but nevertheless are guaranteed to find such approximate winning committees with arbitrarily high probability.
A more formal description of our setting is given in Section~\ref{section:preliminaries}.



Our results,
which are summarized in Table~\ref{table:results},
imply that it is possible to process huge amount of preferences data
(that is, huge amount of voters),
using only small space,
and still, with high probability, find an almost-optimal winning committee.
Since,
as briefly mentioned above,
the voting rules we consider in this paper have applications not only in political settings,
but also in commercial and business settings,
our results naturally have implications to those scenarios as well.
We further discuss the applicability of our results in Section~\ref{section:outlook}.

\subsection{Related Work}

The two most-related papers to our paper are two papers by Bhattacharyya and Dey~\cite{bhattacharyya2015fishing, dey2015sample}.
The first paper~\cite{bhattacharyya2015fishing}
provides an analysis of the space complexity of streaming algorithms for some single-winner voting rules.
The second paper~\cite{dey2015sample}
provides an analysis concerning the number of samples which are sufficient in order to approximately compute the winner under various single-winner voting rules,
and is of relevance to our paper since our algorithms are based on sampling.
Another related paper is that of Filtser and Talmon~\cite{distributedmonitoring}
which provides efficient protocols for winner determination in distributed streams.
We stress that,
while the above-mentioned papers deal with single-winner voting rules,
our paper deals with multiwinner voting rules which select a committee of fixed size.

Another line of work worth mentioning is concerned with developing streaming algorithms for the \textsc{Max Cover} problem.
In the \textsc{Max Cover} problem,
we are given a collection of sets over some universe and a budget $k$,
and the task is to find $k$ sets which cover the largest number of elements.
Approval-CC (see Section~\ref{section:preliminaries}) is equivalent to \textsc{Max Cover}
(to see this, interchange voters by elements and candidates by sets; see also, e.g., \cite{skowron2015fully}).

Thus,
the very recent paper by McGregor and Vu~\cite{mcgregor2016better} is of relevance to us;
specifically, they give an upper bound~\cite[Theorem 10]{mcgregor2016better}
which has some similarities with our Theorem~\ref{theorem:Approval-CC_UB},
and they give a lower bound~\cite[Theorem 20]{mcgregor2016better}.
However, their model of a stream is different than ours,
since the items in their streams are the sets (corresponding to the candidates),
while for us the items are the voters (corresponding to the elements).

In the context of social choice,
there are some further interesting papers to mention.
Conitzer and Sandholm~\cite{conitzer2005communication} study communication complexity of various voting rules;
they do not consider approximations and therefore the communication complexity of their protocols is generally quite high.
Along similar lines,
Chevaleyre et al.~\cite{chevaleyre2011compilation} design communication protocols for situations where the set of candidates might change over time.
Chevaleyre et al.~\cite{chevaleyre2009compiling} study compilation complexity of various voting rules;
roughly speaking,
they divide the electorate into two parts,
and are concerned with the amount of information which one part shall transmit to the other in order to correctly identify a winner.
Xia and Conitzer~\cite{xia2010compilation} extend upon this previously-mentioned paper by considering some further variants
as well as some other voting rules not previously studied.
Finally,
we mention the paper by Conitzer and Sandholm~\cite{con-san:c:strategy-proofness} which is concerned with vote elicitation.

\section{Preliminaries}\label{section:preliminaries}

We provide preliminaries regarding
elections and proportional representation voting rules,
streaming algorithms and vote streams,
and mention some useful results from probability theory.
We denote the set $\{1, \ldots, n\}$ by $[n]$.

\subsection{Proportional Representation}

An \emph{election} $E = (C, V)$
consists of a set of \emph{candidates} $C = \{c_1, \ldots , c_m\}$
and a collection of \emph{voters} $V = (v_1, \ldots , v_n)$,
where each voter is associated with her \emph{vote}.
(For ease of presentation, we refer to the voters as females while the candidates are males.)
In this paper we consider two kinds of elections:
  in \emph{Approval-based} elections,
  the vote of voter $v_i \in V$ is a subset of $C$,
  corresponding to the candidates which this voter \emph{approves};
  in \emph{Borda-based} elections,
  the vote of voter $v_i \in V$ is a total order $\pref_{v_i}$ over $C$.
For Borda-based elections, we write $\pos_v(c)$ to denote the \emph{position} of candidate $c$ in $v$'s preference order
(e.g., if $v$ ranks $c$ on the top position, then $\pos_v(c) = 1$).

Given an election $E = (C, V)$ and an integer $k$,
$k \leq |C|$,
a \emph{committee} $S \subseteq C$ consists of $k$ candidates from $C$.
A \emph{multiwinner voting rule} $\calR$ is a function that returns a set $\calR(E, k)$
of \emph{winning committees} of size $k$ each,
and we say that the committees in $\calR(E, k)$ tie as winners of the election.
To formally define the specific voting rules which we consider in this paper,
namely Chamberlin--Courant and Monroe,
we first discuss assignment functions and satisfaction functions.

\medskip\noindent\textbf{Assignment functions.}\quad
Given an election $E = (C, V)$ and a committee $S \subseteq C$ of size $k$,
a \emph{CC-assignment} function is a function 
$\Phi \colon V \rightarrow S$.
We say that $\Phi(v)$ is the \emph{representative} of voter $v \in V$
and that $v$ is represented by $\Phi(v)$.
An \emph{M-assignment} function is a CC-assignment function
where $\lfloor \frac{n}{k} \rfloor \leq |\Phi^{-1}(c)| \leq \lceil \frac{n}{k} \rceil$
holds for each $c \in S$.
That is,
in an M-assignment,
each committee member represents roughly (i.e., up to rounding) the same number of voters.

\medskip\noindent\textbf{Satisfaction functions.}\quad
Intuitively,
a \emph{satisfaction function} $\gamma : V \times C \to \mathbb{N}$ is a function measuring the satisfaction of a voter $v$
when she is represented by a certain candidate $c$.
For Approval-based elections,
we use the satisfaction function $\gamma \equiv \alpha$ where $\alpha(v, c) = 1$ iff $c$ is approved by $v$,
and $0$ otherwise
(that is, $1$ if $c$ is contained in $v's$ vote; informally, a voter is satisfied only by her approved candidates).
For Borda-based elections,
we use the satisfaction function $\gamma \equiv \beta$ where $\beta(v, c) = m - \pos_v(c)$.

\medskip\noindent\textbf{Chamberlin--Courant and Monroe.}\quad
Given an election $E = (C, V)$,
a size-$k$ committee $S$,
and a CC-assignment function $\Phi$,
we define the \emph{total satisfaction} of the voters in $V$ from the committee $S$ and the CC-assignment $\Phi$ to be:
\[
  \sum_{v \in V}\gamma(\Phi(v)),
\]
where,
for Approval-based elections,
$\gamma$ equals the $\alpha$ satisfaction function described above,
while for Borda-based elections,
$\gamma$ equals the $\beta$ satisfaction function described above. 

\pagebreak

For the Chamberlin--Courant rule,
the \emph{total satisfaction} of the voters in $V$ from a committee $S$
is defined as the maximum total satisfaction of the voters $V$ from the committee $S$ over all possible CC-assignment functions.
The Chamberlin--Courant rule outputs all size-$k$ committees $W$ with the highest total satisfaction.

The Monroe rule is defined similarly,
but where we consider only M-assignment functions;
that is,
the total satisfaction of the voters in $V$ from a committee $S$
is defined as the maximum total satisfaction of the voters $V$ from the committee $S$ over all possible M-assignment functions.
We denote by \emph{Approval-CC} (\emph{Borda-CC}) the Chamberlin--Courant rule for Approval-based (Borda-based) elections,
and by \emph{Approval-M} (\emph{Borda-M}) the Monroe rule for Approval-based (Borda-based) elections.




\begin{table}[t]
\centering
\renewcommand{\arraystretch}{1.5}
\begin{tabular}{ c | c }
Voting rule    & Space complexity                   \\ \hline
Approval-CC    & $O(\epsilon^{-2} k m \log m)$      \\
Borda-CC       & $O(\epsilon^{-2} k^3 m^3 \log m)$  \\
Approval-M     & $O(\epsilon^{-2} k^3 m \log m)$    \\
Borda-M        & $O(\epsilon^{-2}k^3m^5 \log m)$    
\end{tabular}
\caption{Summary of our upper bounds. We list our upper bounds for randomized streaming algorithms
which identify $\epsilon$-approximate winning committees under several proportional representation voting rules.
$k$ denotes the size of the committee while $m$ denotes the number of candidates which participate in the election.}
\label{table:results}
\end{table}

\subsection{Vote Streams}

We assume that the set of candidates $C$ is known,
and that the voters $v_1, \ldots, v_n$ arrive (that is, vote) one at a time.
More formally we might say that at time $t \in [n]$,
voter $v_t$ arrives;
importantly,
each voter arrives only once.

We are interested in randomized algorithms which operate on such vote streams,
and find approximate solutions.
The following definition is crucial for our notion of approximation.

\begin{definition}[$\epsilon$-winning committee]
  A committee of size $k$ is \emph{$\epsilon$-winning}
  if it is either a winning committee,
  or it can become a winning committee by changing at most $\epsilon n$ votes.
\end{definition}  

Specifically,
we require that the committees computed by our streaming algorithm shall be,
with high probability,
$\epsilon$-winning.
Such a streaming algorithm,
which identify,
with high probability,
an $\epsilon$-winning committee,
is said to be an \emph{$\epsilon$-approximate streaming algorithm}.

\begin{definition}[$\epsilon$-approximate streaming algorithm]
  A streaming algorithm is an \emph{$\epsilon$-approximate} streaming algorithm
  if it returns, with high probablty, an $\epsilon$-winning committee.
\end{definition}  

Throughout the paper,
when we say ``with high probability''
we mean with probability $1 - O(1/n)$.
Such a success probability should be sufficient;
as usual in streaming algorithms,
can be further tweaked by repetitions.

Assuming that the number $n$ of voters is huge,
our goal is to devise streaming algorithms whose space complexity do not depend on the number $n$ of voters.
Our algorithms are based on sampling voters;
by a \emph{subset} of an election we mean a subset of the voters.

Let us explain how exactly we sample voters.
Let $n$ be the length of the stream (i.e., the total number of voters),
and suppose that we want to sample $z$ votes from the stream.
Then,
we pick each vote with probability $z / (n \delta)$ for some constant $0 \leq \delta \leq 1$.
By Markov's inequality,
with probability at least $1 - \delta$ is holds that the sample size is at least $z$ (and not much larger).
Hence,
every vote belongs to our sample with probability $z / (n \delta)$ independently of other items.

\subsection{Useful Results from Probability Theory}

Since our algorithms are randomized,
specifically based on sampling a small number of voters,
we make extensive use of the following variant of Hoeffding's inequality,
which upper bounds the probability that the sum of a given set of random variables deviates from its expectation.

\begin{thm}[Hoeffding's inequality~\cite{hoeffding1963probability}]\label{theorem:Hoeffding}
		Let $X_1,..., X_t$ be independent random variables such that $0 \leq X_i \leq m$ for each $i \in [t]$.
		Let $X$ be a random variable such that $X = \sum_{i \in [t]} X_i$.
		Then,
    the following two statements hold.
		\[
			(1) \quad\quad\quad \Pr [X -\mathbb{E}[X] < \epsilon] \le \exp\left(-\frac{2\epsilon^2 t}{m^2}\right)
		\] 
		\[
			(2) \quad\quad\quad \Pr [\mathbb{E}[X] - X < \epsilon] \le \exp\left(-\frac{2\epsilon^2 t}{m^2}\right)
		\]
\end{thm}

For the special case when $m = 1$, Hoeffding's inequality simplifies as follows.
\[
	(1) \quad\quad\quad \Pr [\mathbb{E}[X] - X < \epsilon] \le \exp\left(-{2\epsilon^2 t}\right)
\]

\[
	(2) \quad\quad\quad \Pr [X - \mathbb{E}[X] < \epsilon] \le \exp\left(-{2\epsilon^2 t}\right)
\]

\section{Results}\label{section:results}

Our main results are summarized in Table~\ref{table:results}.
In Section~\ref{section:upperbounds} we describe our upper bounds
while in Section~\ref{section:lowerbounds} we describe our lower bounds.

\subsection{Upper Bounds}\label{section:upperbounds}

We first consider the Approval-CC voting rule,
which is arguably the simplest voting rule we consider in this paper.
The following algorithm is based on sampling a small number of voters.
The proof shows that,
with high probability,
a winning committee for the election corresponding to the sample has fairly high score in the whole election;
specifically, it constitutes an $\epsilon$-approximate winning committee of the whole election.

\begin{theorem}\label{theorem:one}\label{theorem:Approval-CC_UB}
  There is an $\epsilon$-approximate streaming algorithm for Approval-CC
  which uses $O(\epsilon^{-2} k m \log m)$ space.
\end{theorem}

\begin{proof}
The algorithm operates as follows.
We select a sample of $t = 6\epsilon^{-2} k \log m$ voters,
uniformly at random.
Then, we find a winning committee of the sampled voters
(with respect to Approval-CC) and return it as a winning committee for the whole election.
We show that a winning committee of the sampled voters is,
with high probability,
an $\epsilon$-winning committee for the whole election.
Notice that in order to store the votes of $t$ voters,
our algorithm uses $mt$ space, as claimed.

Next we prove that,
with high probability,
our algorithm returns an $\epsilon$-winning committee.
Let $E = (C, V)$ denote the whole election
and let $E_R = (C, V_R)$ denote the sampled election,
where $V_R$ denotes the set of $t$ sampled voters.
Let $S$ be a winning committee in the whole election.
Let $\score_E(S)$ ($\score_{E_R}(S)$) denote the score of $S$ in the whole election (in the sampled election, respectively).

Let us first consider the case where $\score_E(S) < \epsilon n$,
that is,
where there are less than $\epsilon n$ voters being satisfied by $S$.
In this case any committee is $\epsilon$-winning,
thus our algorithm is always correct.
Therefore,
from now on we assume that there are at least $\epsilon n$ voters satisfied by $S$.

The next claim concentrates on the winning committee $S$,
which,
since it is winning in $E$,
has high score in $E$;
the claim shows that,
with high probability,
$S$ also has high score in $E_R$.
The factor $\frac{n}{t}$ is a normalization factor.

\begin{claim}\label{claim:theorem1claim1}
  $\frac{n}{t}\cdot \score_{E_R}(S) \geq \score_{E}(S) - \frac{\epsilon}{2}n$ holds with probability at least $1-m^{-k}$.
\end{claim}

\begin{proof}[of claim~\ref{claim:theorem1claim1}]
For $i \in [t]$,
let $X_i$ be an indicator random variable
such that $X_i = 1$ if the $i$th sampled voter is satisfied by $S$, and $X_i = 0$ otherwise.
Let $X = \sum_{i \in [t]} X_i$.

Since $\score_{E}(S)$ equals the number of voters in the whole election which are satisfied by $S$,
it holds that
\[
  \mathbb{P}[X_i = 1] = \score_{E}(S) / n
\]
for each $i \in [t]$.
Then,
from linearity of expectation,
we conclude that
\[
  \mathbb{E}[X] = \frac{t}{n} \cdot \score_{E}(S).
\]
This means that, in expectation, the score of $S$ in $E_R$ is as claimed;
we use Hoeffding's inequality (see Theorem~\ref{theorem:Hoeffding}) to show concentration, as follows.
\begin{align*}
  \mathbb{P}\left[\frac{n}{t}X < \frac{n}{t}\mathbb{E}[X] - \frac{\epsilon}{2}n\right]
  &=
  \mathbb{P}\left[X < \mathbb{E}[X]  - \frac{\epsilon}{2}t \right] \\
  &\leq
  e^{-2(\frac{\epsilon}{2})^2 6 \epsilon^{-2} k \log m} \\
  &\leq
  m^{-k}.
\end{align*}
(proof of claim~\ref{claim:theorem1claim1}) \end{proof}

Claim~\ref{claim:theorem1claim1} shows that,
with high probability,
a committee with high score in the whole election also gets a relatively high score in the sampled election.
Next we show that,
with high probability,
a committee with low score in the whole election also gets a low score in the sampled election.

\begin{claim}\label{claim:theorem1claim2}
  Let $S'$ be a committee for which it holds that \linebreak
  $\score_{E}(S') \leq (1 - \epsilon) \cdot \score_{E}(S)$.
  Then,
  with probability at least $1 - m^{-2k}$,
  it holds that $\frac{n}{t}\cdot \score_{E_R}(S') \leq  (1 - \epsilon) \cdot$ $\score_{E}(S) + \frac{\epsilon}{2}n$.
\end{claim}

\begin{proof}[of claim~\ref{claim:theorem1claim2}]
Let $S'$ be such that $\score_{E}(S') \leq (1 - \epsilon) \cdot \score_{E}(S)$.
For $i \in [t]$,
let $X_i$ be an indicator random variable
such that $X_i = 1$ if the $i$th sampled voter is satisfied by $S'$, and $X_i = 0$ otherwise.
Let $X = \sum_{i \in [t]} X_i$.
Since $\score_{E}(S')$ equals the number of voters in the whole election which are satisfied by $S'$,
it holds that
\[
  \mathbb{P}[X_i = 1] = \frac{\score_{E}(S')}{n}
\]
for each $i \in [t]$.
Then,
from linearity of expectation,
we conclude that
\[
  \mathbb{E}[X] = \frac{t}{n} \cdot \score_{E}(S') \leq \frac{t}{n} \cdot (1 - \epsilon) \cdot \score_E(S).
\]
This means that, in expectation, the score of $S$ in $E_R$ is as claimed.
Since the $X_i$'s are independent and all of them are bounded,
we use Hoeffding's inequality (see Theorem~\ref{theorem:Hoeffding}) to show concentration, as follows.
\begin{align*}
  \mathbb{P}\left[\frac{n}{t}\cdot X > \frac{n}{t}\cdot \mathbb{E}[X] + \frac{\epsilon}{2}n\right]
  &=
  \mathbb{P}\left[X - \mathbb{E}[X] > \frac{\epsilon }{2}t\right] \\
  &<
  e^{-2 (\epsilon/2)^2 6 \epsilon^{-2} k \log m} \\
  &\leq
  m^{-2k}.
\end{align*}
(proof of claim~\ref{claim:theorem1claim2})\end{proof}

Since there are at most ${m \choose k}\le m^k$ committees,
and therefore at most $m^k$ committees $S'$ for which $\score_{E}(S') \leq (1 - \epsilon) \cdot \score_{E}(S)$ holds
(and these are exactly the committees which are not $\epsilon$-winning),
we can apply union bound on the result proved in Claim~\ref{claim:theorem1claim2},
to get that with high probability,
the score of $S$ in $E_R$ is strictly higher than the score of any committee $S'$ which is not $\epsilon$-winning.
Thus,
our algorithm returns,
with high probability,
an $\epsilon$-winning committee.
\qed\end{proof}

It turns out that it is possible to extend the sampling-based streaming algorithm
described in the proof of Theorem~\ref{theorem:Approval-CC_UB} to work also for Borda-CC,
albeit with some increase of the space complexity.
Informally,
the increase of the space complexity is because the proof needs to take care
for the fact that the score difference induced by a single voter is greater in Borda-CC than it is in Approval-CC:
  while in Approval-CC,
  the satisfaction of a voter from a committee is either $0$ or $1$,
  in Borda-CC it is anything between $0$ to $m - 1$.

\begin{theorem}\label{theorem:Borda-CC-UB}
  There is an $\epsilon$-approximate streaming algorithm for Borda-CC
  which uses $O(\epsilon^{-2} k^3 m^3 \log m)$ space.
\end{theorem}

\begin{proof}
Let $t = 10 \epsilon^{-2} k m^2$.
Similarly in spirit to the algorithm presented in the proof of Theorem~\ref{theorem:one},
our algorithm samples $k^2 t$ voters,
select a winning committee in the sampled election,
and declares it as an $\epsilon$-winning committee for the whole election.
Since storing the vote of each sampled voter takes $m \log m$ space,
we get the claimed space complexity.
Next we prove the correctness of the algorithm.

Fix an election $E$, a committee $S$, a committee member $c$,
and consider a voter $v$.
We define the \emph{score given to $c$ by $v$ with respect to $S$},
denoted by $\score_E^{v, S}(c)$ to be the Borda-score of $c$ in the preference order of $v$,
if,
among the candidates of $S$,
$c$ is the representative of $v$;
that is, if, among the candidates of $S$, $v$ ranks $c$ the highest.
We define it to be $0$ otherwise.
Further,
we define the \emph{score of $c$ with respect to $S$},
denoted by $\score_E^{S}(c)$ to be the sum over all voters,
that is,
$\score_E^{S}(c) = \sum_{i \in [n]} \score_E^{{v_i}, S}(c)$.
Further,
as before,
we define $\score_E(S)$ to be the score of $S$,
and,
indeed,
it holds that $\score_E(S) = \sum_{c \in S} \score_E^{S}(c)$.

We begin by showing that,
fixing a committee $S$ and a committee member $c$,
it is possible to estimate the score of $c$ with respect to $S$ by sampling $t$ voters.
Let $E$ denote the whole election,
and let $E_R$ denote the sampled election,
containing $t$ voters (where $t$ is as defined in the beginning of the current theorem's proof)
chosen uniformly at random from $E$.
The following claim shows that with high probability
the sampled election roughly preserves the score of any committee.

\begin{claim}\label{claim:theorem2claim1}
  Let $S$ be a committee and $c$ a committee member.
  Then, $|\frac{n}{t} \cdot\score_{E_R}^{S}(c) - \score_{E}^{S}(c)| \leq \epsilon n / 2$ holds with probability at least $1 - 1/m^{3k}$,
  where $E_R$ is obtained by sampling each voter in $E$ independently with probability $t/n$.
\end{claim}

\begin{proof}[of claim~\ref{claim:theorem2claim1}]
For the committee $S$ and the committee member $c$,
we define a random variable $X_i$, $i \in [t]$,
such that $X_i = \score_{E_R}^{v_i, S}(c)$,
where $v_i$ is the $i$th sampled voter.
It holds that
\[
  \mathbb{E}[X_i] = \frac{1}{n} \score_{E}^{S}(c).
\]

Letting $X = \sum_{i \in [t]} X_i$,
we have the following (from linearity of expectation):
\[
  \mathbb{E}[X] = \frac{t}{n} \score_{E}^{S}(c).
\]
Importantly,
note that the variables $X_i$ have the following properties:
\begin{itemize}

\item
They are independent; this follows since we consider each committee member separately.

\item
They are bounded; specifically, $0 \leq X_i \leq m$ holds for each $i \in [t]$.

\end{itemize}
Utilizing the above two properties,
we can apply a variation of Hoeffding's inequality (see Theorem~\ref{theorem:Hoeffding}) and conclude that:
\begin{align*}
  \mathbb{P}\left[ | \frac{n}{t}\cdot X - \score_{E}^{S}(c) | \geq n\epsilon / 2 \right]
  &= 
  \mathbb{P}[ | X - \mathbb{E}[X] | \geq t\epsilon / 2 ] \\
  &\leq 
  2 e^{-\frac{ 2(\epsilon/2)^2 t}{ (m + 1)^2}} \\
  &\leq
  \frac{1}{m^{3k}},
\end{align*}
where the first inequality follows from Hoeffding's inequality (see Theorem~\ref{theorem:Hoeffding}) and last inequality follows from our definition of the sample size $t$.
(of claim~\ref{claim:theorem2claim1})\end{proof}

Claim~\ref{claim:theorem2claim1} shows that by sampling $t$ voters,
we get a good estimation for the score of a candidate with respect to some committee.
Let $E$ denote the whole election,
and let $E_R$ denote the sampled election,
containing $k^2 t$ voters (where $t$ is as defined in the beginning of the current theorem's proof)
chosen uniformly at random from $E$.
Next we show that, by sampling $k^2 t$ voters,
we get a good estimation for the score of a committee.

\begin{claim}\label{claim:theorem2claim2}
  Let $S$ be a committee.
  Then,
  $| \frac{n}{t}\cdot\score_{E_R}(S) -  \score_{E}(S) | \leq n \epsilon / 2$ holds with probability at least $1 - 1/m^{k}$,
  where $E_R$ is obtained by sampling each voter in $E$ independently with probability $k^2 t/n$.
\end{claim}

\begin{proof}[of claim~\ref{claim:theorem2claim2}]
Let $S$ be a committee containing the committee members $c_1, \ldots, c_k$.
For each $j \in [k]$,
we apply Claim~\ref{claim:theorem2claim1} on the committee $S$ and the committee member $c_j$ with $\epsilon' = \epsilon / k$.
Let us denote the random variable containing the estimated score of committee member $c_j$ with respect to committee $S$ by
$Y_j$; that is, $Y_j$ is the estimated value of $\score_{E}^{S}(c_j)$, therefore, $Y_j = \score_{E_R}^{S}(c_j)$
using Claim~\ref{claim:theorem2claim1}.
Let $Y = \sum_{j \in [k]} Y_j$.
Since $\score_{E}(S) = \sum_{j \in [k]} \score_{E}^{S}(c_j)$,
and from linearity of expectation,
it follows that 
\[
  \mathbb{E}[Y] = \frac{t}{n}\cdot\score_{E}(S).
\]
Further, we have that:
\begin{align*}
  \mathbb{P}\big[ | \frac{n}{t}\cdot Y &- \score_{E}(S) | \geq n \epsilon / 2 \big] \\
  &\leq 
  \mathbb{P}\left[ \Sigma_{j \in [k]} | Y_j - \mathbb{E}[ Y_j ] | \geq n k \epsilon' / 2 \right] \\
  &\leq
  \sum_{j \in [k]}\left( \mathbb{P}\left[ | Y_j - \mathbb{E}[ Y_j ] | \geq n \epsilon' / 2 \right] \right) \\
  &\leq
  \frac{k}{m^{2k}} \\
  &\leq 
  \frac{1}{m^k},
\end{align*}
where the first inequality follows from the definitions of $Y$ and $\epsilon'$,
the second inequality follows from applying a union bound over the committee members $c_1, \ldots, c_k$,
and the third inequality follows from Claim~\ref{claim:theorem2claim1}.
(of claim~\ref{claim:theorem2claim2})\end{proof}

Finally,
building upon Claim~\ref{claim:theorem2claim2},
we apply union bound on all ${m \choose k}$ committees of size $k$.
Following this union bound,
we conclude that,
with high probability,
the algorithm returns an $\epsilon$-winning committee.
\qed\end{proof}

%
We mention that the result described in Theorem~\ref{theorem:Borda-CC-UB} transfers to all scoring rules,
albeit with some increase of the space complexity.
That is,
careful analysis of the proof of Theorem~\ref{theorem:Borda-CC-UB} reveals that,
since we can upper bound the values of the random variables $X_j$
by $m$,
it follows that we can apply Hoeffding's inequality (see Theorem~\ref{theorem:Hoeffding}),
which causes an increase of the space complexity by a multiplicative factor of $m^2$,
compared to the space complexity that we get for Approval-CC.

Considering any normalized scoring vector $(\alpha_1, \alpha_2, \ldots, \alpha_m)$ with $\alpha_1 \geq \ldots \geq \alpha_m$
such that $\alpha_1$ is the value given by a voter to her first-choice candidate,
and following the same reasoning as described above,
we see that applying Hoeffding's inequality (see Theorem~\ref{theorem:Hoeffding}) causes an increase of the space complexity by a multiplicative factor of $\alpha_1^2$,
compared to the space complexity we get for Approval-CC.
Specifically,
the resulting space complexity is $O(\epsilon^{-2}k^3 \alpha_1^2 m \log m)$. We know that scoring rules remain unchanged if we multiply every $\alpha_i$ by any constant $\lambda>0$ and/or add any constant $\mu$. Hence, we can assume without loss of generality that for any score vector $\alpha$, there exists a $j$ such that $\alpha_j - \alpha_{j+1}=1$ and $\alpha_k = 0$ for all $k>j$. We call such an $\alpha$ a normalized score vector. 
%

\medskip

Next we move on to consider Monroe (M),
beginning with the arguably simpler case of Approval-M.
Our algorithm is again based on sampling a small number of voters and computing a winning committee for them.
The analysis is more involved,
since we cannot consider all assignments,
but only M-assignments.
A naive analysis would apply union bound on all M-assignments,
but since there are $O(k^n)$ such assignments,
we would get linear space in the number of voters,
which would be too much.
Fortunately,
we can do better,
building upon some structural observations,
as we show next.

\begin{theorem}\label{theorem:approval-m}
  There is an $\epsilon$-approximate streaming algorithm for Approval-M
  which uses $O(\epsilon^{-2} k^3 m \log m)$ space.
\end{theorem}

\begin{proof}
The overall idea is to consider any committee $S$ with its optimal assignment $A^*$.
We will show that,
with high probability,
with respect to $S$,
the score of the assignment $A^*$ on a sampled election 
is close to being the actual score of the committee $S$ on the sampled election.
The theorem would then follows by union bound over all ${m \choose k}$ committees.

More specifically,
for each committee $S$ together with its optimal assignment $A^*$,
we define a \emph{preserving subset} to be a subset $E_P$ of the election $E$ such that,
for each committee member $c \in S$,
the fraction of voters assigned to $c$ which are satisfied by $c$,
as well as the fraction of voters assigned to $c$ which are not satisfied by $c$,
is preserved.
Formally,
we define a preserving subset as follows.

\begin{definition}[\emph{preserving subset}]
Let $S$ be a committee,
let $A^*$ be its optimal assignment,
and let $E_P$ be a subset of the election $E$.
Let $\smiley_{E_P}^{A^*}(c_i)$ denote the set of voters in $E_P$ which are assigned to $c_i$ by $A^*$ and are satisfied by $c_i$
(that is, it holds that $c_i \in v$),
and let $\frownie_{E_P}^{A^*}(c_i)$ denote the set of voters in $E_P$ which are assigned to $c_i$ by $A^*$ and are not satisfied by $c_i$
(that is, it holds that $c_i \notin v$).
Then,
a subset $E_P$ of the election $E$ is a \emph{preserving subset} if for each $c_i \in S$ it holds that
\[
  (1)\quad |\smiley_{E_P}^{A^*}(c_i)| = \frac{|E_P|}{|E|} \cdot |\smiley_{E}^{A^*}(c_i)|
\]
and that
\[
  (2)\quad |\frownie_{E_P}^{A^*}(c_i)| = \frac{|E_P|}{|E|} \cdot |\frownie_{E}^{A^*}(c_i)|,
\]
\end{definition}

That is,
a preserving subset is a subset of the voters of some given election
which,
with respect to the optimal assignment of a given committee,
preserves the (normalized) number of voters assigned to each candidate and are satisfied (unsatisfied) by it.
Next we show that,
for each committee $S$,
with high probability a random subset containing $t = O(\epsilon^{-2} k^3 \log m)$ is close to being a preserving subset.

\begin{claim}\label{lemma:am2}
  Let $E_R$ be a random subset of voters from $E$ obtained by sampling each voter independently at random with probability $t / n$.
  Then,
  for each committee $S$,
  with probability at least $1 - m^{-2k}$,
  it holds that there exists a preserving subset $E_P$ which can be obtained from $E_R$ by changing the vote of at most $\epsilon t$ voters.
\end{claim}

\begin{proof}[of claim~\ref{lemma:am2}]
It suffices to show that,
for each $c_i \in S$,
it holds that $\smiley_{E}^{A^*}(c_i) = \frac{n}{t} \smiley_{E_R}^{A^*}(c_i) \pm \frac{\epsilon n}{2k}$
and also it holds that $\frownie_{E}^{A^*}(c_i) = \frac{n}{t} \frownie_{E_R}^{A^*}(c_i) \pm \frac{\epsilon n}{2k}$,
since then,
the fraction of each of the $k$ sets $\smiley_{E}^{A^*}(c_i)$ and each of the $k$ sets $\frownie_{E}^{A^*}(c_i)$ can preserve its respective fraction
by changing the votes of at most $\frac{\epsilon n}{2k}$ voters.

Since each voter is sampled with probability $t / n$,
we have that 
\[
  \mathbb{E}[\smiley_{E_R}^{A^*}(c_i)] = \frac{t}{n} \smiley_{E}^{A^*}(c_i).
\]

Since each voter is sampled independently,
we can apply Hoeffding's inequality (see Theorem~\ref{theorem:Hoeffding}),
to have the following.
\[
  \mathbb{P} \left[ | \smiley_{E_R}^{A^*}(c_i) - \mathbb{E}[ \smiley_{E_R}^{A^*}(c_i) ]| \geq \frac{\epsilon t}{2k} \right] \leq
  2 e^{-\frac{2t \epsilon^2}{4k^2}} = O(m^{-2k})
\]
and
\[
  \mathbb{P} \left[ | \frownie_{E_R}^{A^*}(c_i) - \mathbb{E}[ \frownie_{E_R}^{A^*}(c_i) ]| \geq \frac{\epsilon t}{2k} \right] \leq
  2 e^{-\frac{2t \epsilon^2}{4k^2}} = O(m^{-2k}).
\]
Thus, we are done.
(of claim~\ref{lemma:am2})\end{proof}

Next we show that, for each committee $S$,
its optimal assignment $A^*$ in $E$ is also an optimal assignment in any preserving subset $E_P$ of $E$.
Notice that the following claim is not probabilistic but combinatorial.

\begin{claim}\label{lemma:am1}
  Let $S$ be a committee, $A^*$ be its optimal assignment,
  and $E_P$ be a preserving subset of $E$.
  Then,
  the restriction of $A^*$ to $E_P$ is an optimal assignment for $S$ in $E_P$.
\end{claim}

\begin{proof}[of claim~\ref{lemma:am1}]
Intuitively,
if there was a better assignment $A_P$ than $A^*$ for $S$ in $E_P$,
then we could change $A^*$ accordingly and get a better assignment for $S$ in $E$,
contradicting the optimality of $A^*$ for $S$ in $E$.

More formally,
let $S$ be a committee, $A^*$ be its optimal assignment,
and $E_P$ be a preserving subset of $E$.
Towards a contradiction,
assume that there is an assignment $A_P \neq A^*$ such that $\score_{E_P}^{A_P}(S) > \score_{E_P}^{A^*}(S)$.
Consider $\bar{E}_P = E \setminus E_P$ and notice that, since $E_P$ is a preserving subset of $E$,
it also holds that $\bar{E}_P$ is a preserving subset of $E$,
and we have that
\begin{align*}
  \score_{E}^{A^*}(S) &= \frac{|E_P|}{|E|} \cdot \score_{E}^{  A^*}(S) + \frac{|\bar{E}_P|}{|E|} \cdot \score_{E}^{A^*}(S) \\
                    &= \score_{E_P}^{A^*}(S) + \score_{\bar{E}_P}^{A^*}(S) \\
                    &< \score_{E_P}^{A_P} (S) + \score_{\bar{E}_P}^{A^*}(S).
\end{align*}

Since $A_P$ does not violate the Monroe property,
we have constructed a better assignment for $S$ in $E$,
contradicting the optimality of $A^*$ for $S$ in $E$.
(of claim~\ref{lemma:am1})\end{proof}

Building upon the last two claims proven above,
the following claim shows that,
for each committee $S$,
with high probability,
the score of its optimal assignment $A^*$ on $E$
is a good estimator for its score on the sampled election $E_R$.

\begin{claim}\label{lemma:am0}
  For each committee $S$ and its optimal assignment $A^*$,
  with probability at least $1 - m^{-2k}$ it holds that:
  \[
    \score_{E_R}(S) + \epsilon t \geq \score_{E_R}^{A^*}(S) \geq \score_{E_R}(S) - \epsilon t.
  \]
\end{claim}

\begin{proof}[of claim~\ref{lemma:am0}]
Combining the last two claims,
we have that,
with high probability,
there exists a preserving subset $E_P$,
obtained from the sampled election by changing at most $\epsilon t$ voters.
Consider the preserving subset $E_P$ which is obtained from the sampled election $E_R$ by changing at most $\epsilon t$ voters.

By the first claim,
we have that the assignment $A^*$ is optimal for $S$ on $E_P$.
Consider any other assignment.
Since $A^*$ is optimal for $S$ on $E_P$ and $E_P$ is $\epsilon$-close to $E_R$,
the two inequalities hold,
since $\epsilon$ bounds the score difference between $A^*$ on $E_R$ and any other assignment.
(of claim~\ref{lemma:am0})\end{proof}

Following the last claim,
we have that,
for each committee $S$,
a random sample is indeed a good estimator for the score of $S$.
Then,
the claim follows by union bound over all possible ${m \choose k}$ committees.
\qed\end{proof}

\begin{theorem}
  There is an $\epsilon$-approximate streaming algorithm for Borda-M
  which uses $O(\epsilon^{-2} k^3 m^5 \log m)$ space.
\end{theorem}

\begin{proof}[sketch]
The idea of the proof is very similar to Approval-M,
when we take into account the following two differences.

The first difference is that,
instead of only two blocks for each committee member,
namely the $\smiley$ block and the $\frownie$ block,
in Borda-M we shall consider $m$ blocks for each committee member,
where a voter $v$ is assigned to the $l$th block (for $l \in [m]$)
of committee member $c$ if $v$ is represented by $c$ and the satisfaction of $v$ from $c$ is $j$.

The second difference is that we shall bound the difference between the actual score of a committee
and its score in the sampled election differently;
specifically,
we have that
\[
  \score_{E_R}(S) + \epsilon m t \geq \score_{E_R}^{A^*}(S) \geq \score_{E_R}(S) - \epsilon m t,
\]
since each voter whose vote is changed can increase or decrease the score of each committee by $O(m)$ and not only by $O(1)$ as for Approval-M.

The proof then follows similar lines as the proof given for Approval-M (see Theorem~\ref{theorem:approval-m}),
but the space complexity increases.
Specifically,
the first difference described above causes the space complexity to multiply by a factor of $O(m^2)$,
since we shall consider those $m$ blocks (instead of only $2$) and take into account that the error can multiply by $m$.
Similarly,
the second difference described above causes the space complexity to multiply by another factor of $O(m^2)$,
since we shall increase the size of the sample to account for the increased score difference.
\end{proof}




\subsection{Lower Bounds}\label{section:lowerbounds}

In this section we prove two types of lower bounds which complement our algorithms.
We begin by showing that any streaming algorithm shall use space which is linear in the number $m$ of candidates.

\begin{theorem}\label{theorem:lowerbound1}
  There is an $\epsilon > 0$ such that any $\epsilon$-approximate streaming algorithm for Approval-CC or Approval-M
  needs $\Omega(m)$ space.
\end{theorem}

\begin{proof}
We reduce from the \textsc{Set Disjointness} problem in communication complexity.
In the \textsc{Set Disjointness} problem,
there is a set of elements $U = x_1, \ldots, x_u$,
and two players,
Alice and Bob.
Alice is given a subset $A \subseteq U$ and Bob is given a subset $B \subseteq U$.
Then,
Alice sends a message to Bob,
and Bob has to decide whether $A \cap B = \emptyset$,
in which case Bob shall accept;
otherwise,
that is,
if there is some index $i \in [u]$ such that $x_i \in A \cap B$,
then Bob shall reject.
It is known that Alice shall send $\Omega(u)$ bits in order for Bob to be correct with high probability~\cite{kalyanasundaram1992probabilistic}.

We first describe the reduction for \textsc{Approval-CC};
that is,
given an instance of \textsc{Set Disjointness},
we construct a vote stream for \textsc{Approval-CC},
as follows.
we create an election with $u + 1$ candidates,
where for each $x_i$ ($i \in [u]$) we create a corresponding candidate $c_i$,
and we have another candidate $d$.
Then,
Alice inserts two voters, $v_1$ and $v_1'$, to the vote stream,
where both $v_1$ and $v_1'$ are approving the candidates corresponding to the elements in $A$
(that is, $v_1 = v_1' = \{c_i : x_i \in A\}$.
Then,
Bob inserts two voters, $v_2$ and $v_2'$, to the vote stream,
where,
similarly,
$v_2 = v_2' = \{c_i : x_i \in B\}$.
Finally,
Bob inserts three voters,
$v_3, v_4, v_5$,
all of which approve only the candidate $d$.
This finishes the description of the reduction.

For example,
letting $U = \{x_1, x_2, x_3\}$ (thus, $u = 3$),
$A = \{x_2\}$, and $B = \{x_1, x_2\}$,
we will have that $v_1$ and $v_1'$ both approve $c_2$,
$v_2$ and $v_2'$ both approve $c_1$ and $c_2$,
and $v_3$, $v_4$, and $v_5$ all approve $d$.

We assume,
towards a contradiction,
that there is a streaming algorithm for Approval-CC which uses $o(m)$ space.
We use that algorithm with $k = 1$ and $\epsilon = 1/7$.
Notice that if $A \cap B = \emptyset$,
then each candidate $c_i$ covers at most $2$ voters,
while if there is some index $i \in [u]$ such that $x_i \in A \cap B$,
then the candidate $c_i$ covers $4$ voters.
Irrespectively,
the candidate $d$ covers $3$ voters.
Thus,
the streaming algorithm would declare $d$ as the winner if and only if $A$ and $B$ are disjoint,
contradicting the lower bound for \textsc{Set Disjointness}.

As for Approval-M,
notice that in the reduction described above the size $k$ of the committee is $1$.
In this case,
Approval-CC and Approval-M are equivalent,
thus the reduction transfers to Approval-M as it is.
\end{proof}

It turns out that with some modifications,
the reduction described in the proof of Theorem~\ref{theorem:lowerbound1}
can be made to work also for Borda-CC and Borda-M.

\begin{theorem}\label{theorem:lowerbound1}
  There is an $\epsilon > 0$ such that any $\epsilon$-approximate streaming algorithm for Borda-CC or Borda-M
  needs $\Omega(m)$ space.
\end{theorem}

\begin{proof}
We again reduce from \textsc{Set Disjointness}
where Alice (Bob) is given a subset $A \subseteq U$ ($B \subseteq U$),
for $U = \{x_1, \ldots, x_u$\},
and Alice and Bob shall decide together whether $A \cap B = \emptyset$
(see the proof of Theorem~\ref{theorem:lowerbound1} for a more detailed description of \textsc{Set Disjointness}).

We describe first the reduction for \textsc{Borda-CC};
that is,
given an instance of \textsc{Set Disjointness},
we construct a vote stream for \textsc{Borda-CC},
as follows.
We create an election with $4u + 1$ candidates,
where for each $x_i$ ($i \in [u]$) we create a corresponding candidate $c_i$;
we have another candidate $d$;
and another $3u$ dummy candidates $d_1, \ldots, d_{3u}$.

Corresponding to her set $A$,
Alice inserts one voter $v_1$ to the vote stream,
ranking first those $|A|$ candidates $c_i$ which correspond to the elements $x_i$ in $A$,
then $u - |A|$ dummy candidates $d_1, \ldots, d_{u - |A|}$,
then $d$,
then the remaining $2u + |A|$ dummy candidates $d_{u - |A| + 1}, \ldots, d_{3u}$,
and ranking last those $u - |A|$ candidates $c_i$ which correspond to the elements $x_i$ not in $A$.
Bob behaves quite similarly,
by inserting one voter $v_2$ to the vote stream,
ranking first those $|B|$ candidates $c_i$ which correspond to the elements $x_i$ in $B$,
then $u - |A|$ dummy candidates $d_{3u}, \ldots, d_{3u - |B| + 1}$ (notice the change of order of the dummy candidates with respect to $v_1$),
then $d$,
then the remaining $2u + |B|$ dummy candidates $d_{2u - |B|}, \ldots, d_{1}$ (notice again the change of order),
and ranking last those $u - |B|$ candidates $c_i$ which correspond to the elements $x_i$ not in $B$.
This finishes the description of the reduction.
For example, letting $U = \{x_1, x_2, x_3\}$ (thus, $u = 3$),
$A = \{x_2\}$, and $B = \{x_1, x_2\}$,
we will have that
$v_1 : c_2 \pref d_1 \pref d_2 \pref d \pref d_3 \pref d_4 \pref d_5 \pref d_6 \pref d_7 \pref d_8 \pref d_9 \pref c_1 \pref c_3$
and
$v_2 : c_1 \pref c_2 \pref d_9 \pref d \pref d_8 \pref d_7 \pref d_6 \pref d_5 \pref d_4 \pref d_3 \pref d_2 \pref c_1 \pref c_3$.

We argue that $d$ is a Borda winner in the reduced election if and only if $A \cap B = \emptyset$.
Let us denote the Borda score of a candidate $c$ in the election containing the voters $v_1$ and $v_2$ by $s(c)$.
For the dummy candidates we have that $s(d_i) \leq 5u$ (for any $i \in [3u]$);
this can be seen by observing that the dummy candidates achieve maximum score in the extreme case where $A = B = \emptyset$,
in which $d_i$ is getting $4u - i$ points from $v_1$ and another $u + i$ points from $v_2$.

Now,
consider a candidate $c_i$ corresponding to an element $x_i$ which appears only in one of the sets, either $A$ or $B$;
without loss of generality, let $c_i$ be a candidate corresponding to an element $x_i$ such that $x_i \in A$ and $x_i \notin B$.
Then,
we have that $c_i$ gets at most $4u$ points from $v_1$ and at most $u - 1$ points from $v_2$.
Thus, we conclude that $s(c_i) \leq 5u - 1$.
Similarly,
consider a candidate $c_i$ corresponding to an element $x_i$ which appears both in $A$ and $B$.
Then,
we have that $c_i$ gets at least $3u + 1$ points from each of $v_1$ and $v_2$.
Thus,
we conclude that $s(c_i) \geq 6u + 2$.

Finally,
notice that,
irrespective of the contents of $A$ and $B$,
it holds that $s(d) = 6u$ .
Therefore,
following the computation described in the last paragraph,
we conclude that $d$ is a Borda winner if and only if $A$ and $B$ are disjoint.
So,
assuming,
towards a contradiction,
that there is a streaming algorithm for Borda-CC which uses $o(m)$ space,
we use that algorithm with $k = 1$ and $\epsilon = 1/3$.
Since the streaming algorithm would declare $d$ as the winner if and only if $A$ and $B$ are disjoint,
it would contradict the lower bound for \textsc{Set Disjointness}.

As for Borda-M,
notice that in the reduction described above the size $k$ of the committee is $1$.
In this case,
Borda-CC and Borda-M are equivalent,
thus the reduction transfers to Borda-M as it is.
\end{proof}

We continue by observing the following lower bound, with respect to the required approximation $\epsilon$
(notice that the following theorem is also a corollary of~\cite[Theorem 10]{bhattacharyya2016optimal}).

\begin{theorem}
  For any $\epsilon > 0$,
  any $\epsilon$-approximate streaming algorithm for Approval-CC
  needs $\Omega(\epsilon^{-1})$ space.
\end{theorem}

\begin{proof}
We reduce from the \textsc{$\ell_1$-Heavy Hitters} problem,
which,
given a stream containing $n$ items,
each item is of one type out of $m$ item types,
an approximation parameter $\epsilon$, and a further parameter $\phi$,
asks for returning all items which occur at least $\phi n$ times,
while not returning any item which occurs less than $(\phi - \epsilon)n$ times.
A lower bound of $O(\epsilon^{-1})$ is known for \textsc{$\ell_1$-Heavy Hitters}~\cite{bhattacharyya2016optimal}.

Given an instance of \textsc{$\ell_1$-Heavy Hitters},
we create an instance for Approval-CC, as follows.
For each item type,
we create a candidate.
For each item in the stream,
we create a voter approving only the candidate corresponding to its item type.
We set $k = 1$, keep the same $\epsilon$, and set $\phi = 1/2$.
This finishes the description of the reduction.
Correctness and space complexity follows immediately.
\end{proof}

The reader might notice that the lower bounds presented in this section are not tight.
We leave the task of closing the gap between our upper bounds and lower bounds to future research.


\section{Discussion and Outlook}\label{section:outlook}

We have described streaming algorithms which find approximate winners for several well-known proportional representation multiwinner voting rules.
Below we mention some extensions to our model,
discuss the usefulness of our results,
and mention several avenues for future research.

\mypara{More general models}
In this paper we concentrated on a simple streaming model where
(1) each item in the stream is a voter,
(2) there are no assumptions on the order by which the voters arrive to the stream,
and
(3) the goal is to compute an approximate winner at the end of the stream.

There are other relevant models, which we mention below.

\begin{itemize}

\item In the \emph{sliding windows} model, the goal is to compute an approximate winner
with respect to the last $t$ elements in the stream, for some given $t$.
Since our streaming algorithms are based on sampling,
and sampling from a sliding window can be done efficiently~\cite{babcock2002sampling},
our streaming algorithms extend to this model as well.
This model is useful for identifying emerging trends.

\item It is possible to use our streaming algorithms not only to compute an approximate winner at the end of the stream,
but,
since they are based on sampling,
they can be used to compute an approximate winner at any time during the stream.

\item Our streaming algorithms extend also to situations where we do not know the number $n$ of the voters a-priori,
as is apparent by a recent result~\cite{bhattacharyya2016optimal},
and since our streaming algorithms are based on sampling.

\item Consider situations where a voter might gradually approve more candidates.
A corresponding stream model might be that each item in the stream is a tuple $(v_i, c_j)$,
where an item $(v_i, c_j)$ means that voter $v_i$ have just decided to approve candidate $c_j$.
Such a stream model might model online shopping websites,
where an item $(v_i, c_j)$ would arrive to the stream whenever the person $v_i$ decided to search for the product $c_j$.
Importantly,
since we can decide at the beginning of the stream which voters to sample,
it follows that our upper bounds also extend to this, more general model.

\end{itemize}

\mypara{Less general models}
It might be interesting to study models where we assume some structure in the stream.
Specifically,
one might consider \emph{uniform streams},
where the voters are not arriving in an arbitrary (possibly adversarial) order,
but in a random order, by choosing a random permutation uniformly at random.
The hope is that for such uniform streams it might be possible to design streaming algorithms with better space complexity.
Indeed,
we believe that,
at least for uniform streams,
there are streaming algorithms with better space complexity for \emph{round-based} voting rules,
such as the greedy versions of Chamberlin--Courant and Monroe~\cite{sko-fal-sli:c:multiwinner}
(in short, one might sample several subelections, and use each subelection for a different round).

Such results would be relevant also for situations without huge number of voters,
but with time constraints; consider the following example (which we thank an anonymous reviewer for suggesting it).
A distinguished speaker is to give the same talk at $k$ different dates,
and,
in order to maximize the total number of attendees,
an online scheduling poll is created in order to decide upon the dates.
The problem is that we have to decide upon the dates very soon,
so we cannot wait for everybody to answer;
our sampling-based streaming algorithms (and possibly even better algorithms assuming stream uniformity)
could tell us how many voters we need in the scheduling poll.

Another restricted model might be to consider restricted domains,
thus not considering all possible elections,
but only those elections which adhere to some restricted domains,
such as single peaked domains and single crossing domains.
It is not clear whether imposing structural constraints on the elections would lower the needed space complexity.

\mypara{Other multiwinner voting rules}
Indeed,
streaming algorithms for other multiwinner voting rules deserve to be studied as well.
We specifically mention Single Transferable Vote (STV) which also aims at proportional representation.
Naturally,
there are other multiwinner voting rules which do not aim at proportional representation;
we mention
$k$-best rules,
committee scoring rules,
and various extensions to Condorcet consistent voting rules,
as some important families of multiwinner voting rules.

\bibliographystyle{abbrv}
\bibliography{bib}

\begin{thebibliography}{10}

\bibitem{babcock2002sampling}
B.~Babcock, M.~Datar, and R.~Motwani.
\newblock Sampling from a moving window over streaming data.
\newblock In {\em Proceedings of SODA '02}, pages 633--634, 2002.

\bibitem{bet-sli-uhl:j:mon-cc}
N.~Betzler, A.~Slinko, and J.~Uhlmann.
\newblock On the computation of fully proportional representation.
\newblock {\em Journal of Artificial Intelligence Research}, 47:475--519, 2013.

\bibitem{bhattacharyya2015fishing}
A.~Bhattacharyya and P.~Dey.
\newblock Fishing out winners from vote streams.
\newblock {\em arXiv preprint arXiv:1508.04522}, 2015.

\bibitem{bhattacharyya2016optimal}
A.~Bhattacharyya, P.~Dey, and D.~P. Woodruff.
\newblock An optimal algorithm for $l1$-heavy hitters in insertion streams and
  related problems.
\newblock In {\em Proceedings of PODS '16}, pages 385--400, 2016.

\bibitem{cha-cou:j:cc}
B.~Chamberlin and P.~Courant.
\newblock Representative deliberations and representative decisions:
  Proportional representation and the {B}orda rule.
\newblock {\em American Political Science Review}, 77(3):718--733, 1983.

\bibitem{chevaleyre2011compilation}
Y.~Chevaleyre, J.~Lang, N.~Maudet, and J.~Monnot.
\newblock Compilation and communication protocols for voting rules with a
  dynamic set of candidates.
\newblock In {\em Proceedings of TARK '11}, pages 153--160, 2011.

\bibitem{chevaleyre2009compiling}
Y.~Chevaleyre, J.~Lang, N.~Maudet, and G.~Ravilly-Abadie.
\newblock Compiling the votes of a subelectorate.
\newblock In {\em Proceedings of IJCAI '09}, pages 97--102, 2009.

\bibitem{con-san:c:strategy-proofness}
V.~Conitzer and T.~Sandholm.
\newblock Vote elicitation: {Complexity} and strategy-proofness.
\newblock In {\em Proceedings of AAAI '12}, pages 392--397, 2002.

\bibitem{conitzer2005communication}
V.~Conitzer and T.~Sandholm.
\newblock Communication complexity of common voting rules.
\newblock In {\em Proceedings of EC' 05}, pages 78--87, 2005.

\bibitem{dey2015sample}
P.~Dey and A.~Bhattacharyya.
\newblock Sample complexity for winner prediction in elections.
\newblock In {\em Proceedings of AAMAS '15}, pages 1421--1430, 2015.

\bibitem{elkind2014properties}
E.~Elkind, P.~Faliszewski, P.~Skowron, and A.~Slinko.
\newblock Properties of multiwinner voting rules.
\newblock In {\em Proceedings of AAMAS '14}, pages 53--60, 2014.

\bibitem{clusteringpaper}
P.~Faliszewski, A.~Slinko, K.~Stahl, and N.~Talmon.
\newblock Achieving fully proportional representation by clustering voters.
\newblock In {\em Proceedings of AAMAS '16}, pages 296--304, 2016.

\bibitem{distributedmonitoring}
A.~Filtser and N.~Talmon.
\newblock Distributed monitoring of election winners.
\newblock In {\em Proceedings of AAMAS '17}, 2017.
\newblock To appear.

\bibitem{hoeffding1963probability}
W.~Hoeffding.
\newblock Probability inequalities for sums of bounded random variables.
\newblock {\em Journal of the American statistical association},
  58(301):13--30, 1963.

\bibitem{kalyanasundaram1992probabilistic}
B.~Kalyanasundaram and G.~Schintger.
\newblock The probabilistic communication complexity of set intersection.
\newblock {\em SIAM Journal on Discrete Mathematics}, 5(4):545--557, 1992.

\bibitem{bou-lu:c:chamberlin-courant}
T.~Lu and C.~Boutilier.
\newblock Budgeted social choice: From consensus to personalized decision
  making.
\newblock In {\em Proceedings of IJCAI '11}, pages 280--286, 2011.

\bibitem{mcgregor2016better}
A.~McGregor and H.~T. Vu.
\newblock Better streaming algorithms for the maximum coverage problem.
\newblock {\em arXiv preprint arXiv:1610.06199}, 2016.

\bibitem{mon:j:monroe}
B.~Monroe.
\newblock Fully proportional representation.
\newblock {\em American Political Science Review}, 89(4):925--940, 1995.

\bibitem{pro-ros-zoh:j:proportional-representation}
A.~Procaccia, J.~Rosenschein, and A.~Zohar.
\newblock On the complexity of achieving proportional representation.
\newblock {\em Social Choice and Welfare}, 30(3):353--362, 2008.

\bibitem{skowron2015fully}
P.~Skowron and P.~Faliszewski.
\newblock Fully proportional representation with approval ballots:
  Approximating the maxcover problem with bounded frequencies in {FPT} time.
\newblock In {\em Proceedings AAAI '15}, pages 2124--2130, 2015.

\bibitem{skowron2015finding}
P.~Skowron, P.~Faliszewski, and J.~Lang.
\newblock Finding a collective set of items: from proportional
  multirepresentation to group recommendation.
\newblock In {\em Proceedings of AAAI '15}, pages 2131--2137, 2015.

\bibitem{sko-fal-sli:c:multiwinner}
P.~Skowron, P.~Faliszewski, and A.~Slinko.
\newblock Fully proportional representation as resource allocation:
  {A}pproximability results.
\newblock In {\em Proceedings of IJCAI '13}, pages 353--359, 2013.

\bibitem{xia2010compilation}
L.~Xia and V.~Conitzer.
\newblock Compilation complexity of common voting rules.
\newblock In {\em Proceedings of AAAI '10}, pages 915--920, 2010.

\end{thebibliography}

\end{document}